\documentclass[11pt,onecolumn,draftclsnofoot]{IEEEtran}

\usepackage{amsmath,amssymb} \interdisplaylinepenalty=2500
\usepackage{cite}
\usepackage{graphicx} 

\newtheorem{definition}{Definition}
\newtheorem{theorem}{Theorem}
\newtheorem{proposition}[theorem]{Proposition}
\newtheorem{corollary}[theorem]{Corollary}
\newtheorem{lemma}[theorem]{Lemma}

\newtheorem{myexample}{Example}
\newenvironment{example}{\myexample}{\qed\endmyexample}
\newenvironment{remark}{\textit{Remark: }}{}
\def\qed{\endIEEEproof}

\DeclareMathOperator*{\argmin}{argmin}
\renewcommand{\dim}{\operatorname{\sf dim}\hspace{0.1em}}
\DeclareMathOperator{\rank}{\sf rank\hspace{0.1em}}
\DeclareMathOperator{\wt}{\sf wt\hspace{0.1em}}
\newcommand{\linspan}[1]{\left< #1 \right>}
\newcommand{\slinspan}[1]{\langle #1 \rangle}

\newcommand{\Fq}{\mathbb{F}_q}
\newcommand{\mat}[1]{\begin{bmatrix} #1 \end{bmatrix}}

\newcommand{\dH}{d_{\rm H}}
\newcommand{\dr}{d_{\rm R}}
\newcommand{\ds}{d_{\rm S}}
\newcommand{\di}{d_{\rm I}}

\newcommand{\calC}{\mathcal{C}}

\newcommand{\calE}{\mathcal{E}}

\newcommand{\calP}{\mathcal{P}}

\newcommand{\calU}{\mathcal{U}}
\newcommand{\calV}{\mathcal{V}}

\newcommand{\calX}{\mathcal{X}}
\newcommand{\calY}{\mathcal{Y}}

\newcommand{\magic}{normal}

\title{On Metrics for Error Correction in Network Coding}

\author{Danilo Silva and Frank R. Kschischang
\thanks{This work was supported by CAPES Foundation, Brazil, and by the Natural Sciences and Engineering Research Council of Canada. Portions of this paper were presented at the IEEE Information Theory Workshop, Bergen, Norway, July 2007, and at the 46th Annual Allerton Conference on Communications, Control, and Computing, Monticello, IL, September 2008.}
\thanks{The authors are with The Edward S. Rogers Sr. Department of Electrical and Computer Engineering, University of Toronto, Toronto, ON M5S 3G4, Canada (e-mail: danilo@comm.utoronto.ca, frank@comm.utoronto.ca).}
}

\begin{document}
\maketitle
\thispagestyle{empty}

\begin{abstract}

The problem of error correction in both coherent and noncoherent network coding is considered under an adversarial model. For coherent network coding, where knowledge of the network topology and network code is assumed at the source and destination nodes, the error correction capability of an (outer) code is succinctly described by the rank metric; as a consequence, it is shown that universal network error correcting codes achieving the Singleton bound can be easily constructed and efficiently decoded. For noncoherent network coding, where knowledge of the network topology and network code is not assumed, the error correction capability of a (subspace) code is given exactly by a new metric, called the \emph{injection metric}, which is closely related to, but different than, the subspace metric of K\"otter and Kschischang. In particular, in the case of a non-constant-dimension code, the decoder associated with the injection metric is shown to correct more errors then a minimum-subspace-distance decoder. All of these results are based on a general approach to adversarial error correction, which could be useful for other adversarial channels beyond network coding.

\end{abstract}

\begin{IEEEkeywords}
Adversarial channels, error correction, injection distance, network coding, rank distance, subspace codes.
\end{IEEEkeywords}

\section{Introduction}
\label{sec:introduction}

The problem of error correction for a network implementing linear network coding has been an active research area since 2002 \cite{Cai.Yeung2002,Yeung.Cai2006,Cai.Yeung2006,Yang.Yeung2007,Zhang2008,Jaggi++2008,Kotter.Kschischang2008,Silva++2008,Yang.Yeung2007:RefinedBounds,Yang++2007:ConstructionRefined,Matsumoto2007,Yang++2008:WeightProperties}.
The crucial motivation for the problem is the phenomenon of error propagation, which arises due to the recombination characteristic at the heart of network coding. A single corrupt packet occurring in the application layer (e.g., introduced by a malicious user) may proceed undetected and contaminate other packets, causing potentially drastic consequences and essentially ruling out classical error correction approaches.

In the basic multicast model for linear network coding, a source node transmits $n$ packets, each consisting of $m$ symbols from a finite field $\Fq$. Each link in the network transports a packet free of errors, and each node creates outgoing packets as $\Fq$-linear combinations of incoming packets. There are one or more destination nodes that wish to obtain the original source packets. At a specific destination node, the received packets may be represented as the rows of an $N \times m$ matrix $Y = AX$, where $X$ is the matrix whose rows are the source packets and $A$ is the transfer matrix of the network.
Errors are incorporated in the model by allowing up to $t$ error packets to be added (in the vector space $\Fq^m$) to the packets sent over one or more links. The received matrix $Y$ at a specific destination node may then be written as
\begin{equation}\label{eq:matrix-model}
  Y = AX + DZ
\end{equation}
where $Z$ is a $t \times m$ matrix whose rows are the error packets, and $D$ is the transfer matrix from these packets to the destination.  Under this model, a coding-theoretic problem is how to design an outer code and the underlying network code such that reliable communication (to all destinations) is possible.

This coding problem can be posed in a number of ways depending on the set of assumptions made. For example, we may assume that the network topology and the network code are known at the source and at the destination nodes, in which case we call the system \emph{coherent network coding}. Alternatively, we may assume that such information is unavailable, in which case we call the system \emph{noncoherent network coding}.
The error matrix $Z$ may be random or chosen by an adversary, and there may be further assumptions on the knowledge or other capabilities of the adversary. The essential assumption, in order to pose a meaningful coding problem, is that the number of injected error packets, $t$, is bounded.

Error correction for coherent network coding was originally studied by Cai and Yeung \cite{Cai.Yeung2002,Yeung.Cai2006,Cai.Yeung2006}. Aiming to establish fundamental limits, they focused on the fundamental case $m=1$. In \cite{Yeung.Cai2006,Cai.Yeung2006} (see also \cite{Yang.Yeung2007:RefinedBounds,Yang++2007:ConstructionRefined}), the authors derive a Singleton bound in this context and construct codes that achieve this bound. A drawback of their approach is that the field size required can be very large (on the order of $\binom{|\mathcal{E}|}{2t}$, where $|\mathcal{E}|$ is the number of edges in the network), and no efficient decoding method is given. Similar constructions, analyses and bounds appear also in \cite{Zhang2008,Matsumoto2007,Yang.Yeung2007,Yang++2008:WeightProperties}.

In Section~\ref{sec:coherent-network-coding},
we approach this problem (for general~$m$) under a different framework. We assume the pessimistic situation in which the adversary can not only inject up to $t$ packets but can also freely choose the matrix $D$. In this scenario, it is essential to exploit the structure of the problem when $m>1$. The proposed approach allows us to find a metric---the rank metric---that succinctly describes the error correction capability of a code. We quite easily obtain bounds and constructions analogous to those of \cite{Yeung.Cai2006,Cai.Yeung2006,Yang.Yeung2007:RefinedBounds,Yang++2007:ConstructionRefined,Matsumoto2007}, and show that many of the results in \cite{Yang.Yeung2007,Yang++2008:WeightProperties} can be reinterpreted and simplified in this framework. Moreover, we find that our pessimistic assumption actually incurs no penalty since the codes we propose achieve the Singleton bound of \cite{Yeung.Cai2006}. An advantage of this approach is that it is \emph{universal}, in the sense that the outer code and the network code may be designed independently of each other. More precisely, the outer code may be chosen as any rank-metric code with a good error-correction capability, while the network code can be designed as if the network were error-free (and, in particular, the field size can be chosen as the minimum required for multicast). An additional advantage is that encoding and decoding of properly chosen rank-metric codes can be performed very efficiently~\cite{Silva++2008}.

For noncoherent network coding, a combinatorial framework for error control was introduced by K\"otter and Kschischang in \cite{Kotter.Kschischang2008}. There, the problem is formulated as the transmission of subspaces through an operator channel, where the transmitted and received subspaces are the row spaces of the matrices $X$ and $Y$ in (\ref{eq:matrix-model}), respectively.
They proposed a metric that is suitable for this channel, the so-called subspace distance \cite{Kotter.Kschischang2008}.
They also presented a Singleton-like bound for their metric and subspace codes achieving this bound. The main justification for their metric is the fact that a minimum subspace distance decoder seems to be the necessary and sufficient tool for optimally decoding the disturbances imposed by the operator channel. However, when these disturbances are translated to more concrete terms such as the number of error packets injected, only decoding guarantees can be obtained for the minimum distance decoder of \cite{Kotter.Kschischang2008}, but no converse.
More precisely, assume that $t$ error packets are injected and a general (not necessarily constant-dimension) subspace code with minimum subspace distance $d$ is used. In this case, while it is possible to guarantee successful decoding if $t < d/4$, and we know of specific examples where decoding fails if this condition is not met, a general converse is not known.

In Section~\ref{sec:noncoherent-network-coding}, we prove such a converse for a new metric---called the \emph{injection distance}---under a slightly different transmission model. We assume that the adversary is allowed to arbitrarily select the matrices $A$ and $D$, provided that a lower bound on the rank of $A$ is respected. Under this pessimistic scenario, we show that the injection distance is the fundamental parameter behind the error correction capability of a code; that is, we can guarantee correction of $t$ packet errors \emph{if and only if} $t$ is less than half the minimum injection distance of the code. While this approach may seem too pessimistic, we provide a class of examples where a minimum-injection-distance decoder is able to correct more errors than a minimum-subspace-distance decoder. Moreover, the two approaches coincide when a constant-dimension code is used.

In order to give a unified treatment of both coherent and noncoherent network coding, we first develop a general approach to error correction over (certain) adversarial channels. Our treatment generalizes the more abstract portions of classical coding theory and has the main feature of mathematical simplicity. The essence of our approach is to use a single function---called a \emph{discrepancy function}---to fully describe an adversarial channel. We then propose a distance-like function that is easy to handle analytically and (in many cases, including all the channels considered in this paper) precisely describes the error correction capability of a code. The motivation for this approach is that, once such a distance function is found, one can virtually forget about the channel model and fully concentrate on the combinatorial problem of finding the largest code with a specified minimum distance (just like in classical coding theory). Interestingly, our approach is also useful to characterize the error detection capability of a code.

The remainder of the paper is organized as follows. Section~\ref{sec:preliminaries} establishes our notation and review some basic facts about matrices and rank-metric codes. Section~\ref{sec:adversarial-model} presents our general approach to adversarial error correction, which is subsequently specialized to coherent and noncoherent network coding models. Section~\ref{sec:coherent-network-coding} describes our main results for coherent network coding and discusses their relationship with the work of Yeung et al. \cite{Yeung.Cai2006,Cai.Yeung2006,Yang.Yeung2007}. Section~\ref{sec:noncoherent-network-coding} describes our main results for noncoherent network coding and discusses their relationship with the work of K\"otter and Kschischang \cite{Kotter.Kschischang2008}. Section~\ref{sec:conclusion} presents our conclusions.

\section{Preliminaries}
\label{sec:preliminaries}

\subsection{Basic Notation}

Define $\mathbb{N} = \{0,1,2,\ldots\}$ and $[x]^+ = \max\{x,0\}$. The following notation is used many times throughout the paper. Let $\calX$ be a set, and let $\calC \subseteq \calX$. Whenever a function $d \colon \calX \times \calX \to \mathbb{N}$ is defined, denote
\begin{equation}\nonumber
 d(\calC) \triangleq \min_{x,x' \in \calC:\; x \neq x'}\, d(x,x').
\end{equation}
If $d(x,x')$ is called a ``distance'' between $x$ and $x'$, then $d(\calC)$ is called the \emph{minimum} ``distance'' of $\calC$.

\subsection{Matrices and Subspaces}

Let $\Fq$ denote the finite field with $q$ elements. We use $\Fq^{n \times m}$ to denote the set of all $n \times m$ matrices over $\Fq$ and use $\calP_q(m)$ to denote the set of all subspaces of the vector space $\Fq^m$.

Let $\dim \calV$ denote the dimension of a vector space $\calV$, let $\linspan{X}$ denote the row space of a matrix $X$, and let $\wt(X)$ denote the number of nonzero rows of $X$. Recall that $\dim \linspan{X} = \rank X \leq \wt(X)$.

Let $\calU$ and $\calV$ be subspaces of some fixed vector space. Recall that the sum
  $\calU + \calV = \{u+v \colon u \in \calU,\, v \in \calV\}$
is the smallest vector space that contains both $\calU$ and $\calV$, while the intersection $\calU \cap \calV$ is the largest vector space that is contained in both $\calU$ and $\calV$. Recall also that
\begin{equation}\label{eq:subspace.dimension-equality}
  \dim (\calU + \calV) = \dim \calU + \dim \calV - \dim (\calU \cap \calV).
\end{equation}

The rank of a matrix $X \in \Fq^{n \times m}$ is the smallest $r$ for which there exist matrices $P \in \Fq^{n \times r}$ and $Q \in \Fq^{r \times m}$ such
that $X = PQ$. Note that both matrices obtained in the decomposition are full-rank; accordingly, such a decomposition is called a full-rank decomposition \cite{Rao.Bhimasankaram2000}. In this case, note that, by partitioning $P$ and $Q$, the matrix $X$ can be further expanded as
\begin{equation}\nonumber
  X = PQ = \mat{P' & P''}\mat{Q' \\ Q''} = P'Q' + P''Q''
\end{equation}
where $\rank(P'Q') + \rank(P''Q'') = r$.

Another useful property of the rank function is that, for $X \in \Fq^{n \times m}$ and $A \in \Fq^{N \times n}$, we have \cite{Rao.Bhimasankaram2000}
\begin{equation}\label{eq:bound-rank-product}
  \rank A + \rank X - n \leq \rank AX \leq \min\{\rank A,\,\rank X\}.
\end{equation}

\subsection{Rank-Metric Codes}
\label{sec:rank-metric-codes}

Let $X,Y \in \Fq^{n \times m}$ be matrices. The \emph{rank distance} between $X$ and $Y$ is defined as
\begin{equation}\nonumber
  \dr(X,Y) \triangleq \rank(Y - X).
\end{equation}
It is well known that the rank distance is indeed a metric; in particular, it satisfies the triangle inequality \cite{Gabidulin1985,Rao.Bhimasankaram2000}.

A \emph{rank-metric code} is a matrix code $\calC \subseteq \Fq^{n \times m}$ used in the context of the rank metric. The Singleton bound for the rank metric \cite{Gabidulin1985} (see also \cite{Silva++2008}) states that every rank-metric code $\mathcal{C} \subseteq \Fq^{n \times m}$ with minimum rank distance $\dr(\calC) = d$ must satisfy
\begin{equation}\label{eq:singleton-bound}
|\mathcal{C}| \leq q^{\max\{n,m\} (\min\{n,m\} - d + 1)}.
\end{equation}
Codes that achieve this bound are called \emph{maximum-rank-distance} (MRD) codes and they are known to exist for all choices of parameters $q$, $n$, $m$ and $d \leq \min\{n,m\}$ \cite{Gabidulin1985}.

\section{A General Approach to Adversarial Error Correction}
\label{sec:general-approach}

This section presents a general approach to error correction over adversarial channels. This approach is specialized to coherent and noncoherent network coding in sections \ref{sec:coherent-network-coding} and \ref{sec:noncoherent-network-coding}, respectively.

\subsection{Adversarial Channels}
\label{sec:adversarial-model}


An \emph{adversarial channel} is specified by a finite input alphabet $\calX$, a finite output alphabet $\calY$ and a collection of \emph{fan-out sets} $\calY_x \subseteq \calY$ for all $x \in \calX$. For each input $x$, the output $y$ is constrained to be in $\calY_x$ but is otherwise arbitrarily chosen by an adversary. The constraint on the output is important: otherwise, the adversary could prevent communication simply by mapping all inputs to the same output. No further restrictions are imposed on the adversary; in particular, the adversary is potentially omniscient and has unlimited computational power.

A code for an adversarial channel is a subset\footnote{There is no loss of generality in considering a single channel use, since the channel may be taken to correspond to multiple uses of a simpler channel.} $\calC \subseteq \calX$. We say that a code is \emph{unambiguous} for a channel if the input codeword can always be uniquely determined from the channel output.
More precisely, a code $\calC$ is unambiguous if the sets $\calY_x$, $x \in \calC$, are pairwise disjoint.
The importance of this concept lies in the fact that, if the code is \emph{not} unambiguous, then there exist codewords $x,x'$ that are \emph{indistinguishable} at the decoder: if $\calY_x \cap \calY_{x'} \neq \emptyset$, then the adversary can (and will) exploit this ambiguity by mapping both $x$ and $x'$ to the same output.

A decoder for a code $\calC$ is any function $\hat{x}\colon \calY \to \calC \cup \{ f \}$, where $f \not\in \calC$ denotes a decoding failure (detected error). When $x \in \calC$ is transmitted and $y \in \calY_x$ is received, a decoder is said to be \emph{successful} if $\hat{x}(y) = x$. We say that a decoder is \emph{infallible} if it is successful for all $y \in \calY_x$ and all $x \in \calC$. Note that the existence of an infallible decoder for $\calC$ implies that $\calC$ is unambiguous. Conversely, given any unambiguous code $\calC$, one can always find (by definition) a decoder that is infallible. One example is the exhaustive decoder
\begin{equation}\nonumber
  \hat{x}(y) = \begin{cases}
    x & \text{if $y \in \calY_x$ and $y \not\in \calY_{x'}$ for all $x' \in \calC$, $x'=x$} \\
    f & \text{otherwise}.
  \end{cases}
\end{equation}
In other words, an exhaustive decoder returns $x$ if $x$ is the unique codeword that could possibly have been transmitted when $y$ is received, and returns a failure otherwise.

Ideally, one would like to find a large (or largest) code that is unambiguous for a given adversarial channel, together with a decoder that is infallible (and computationally-efficient to implement).

\subsection{Discrepancy}

It is useful to consider adversarial channels parameterized by an \emph{adversarial effort} $t \in \mathbb{N}$. Assume that the fan-out sets are of the form
\begin{equation}\label{eq:output-sets}
  \calY_x = \{y \in \calY\colon \Delta(x,y) \leq t\}
\end{equation}
for some $\Delta\colon \calX \times \calY \to \mathbb{N}$. The value $\Delta(x,y)$, which we call the \emph{discrepancy} between $x$ and $y$, represents the minimum effort needed for an adversary to transform an input $x$ into an output $y$. The value of $t$ represents the maximum adversarial effort (maximum discrepancy) allowed in the channel.

In principle, there is no loss of generality in assuming (\ref{eq:output-sets}) since, by properly defining $\Delta(x,y)$, one can always express any $\calY_x$ in this form. For instance, one could set $\Delta(x,y) = 0$ if $y \in \calY_x$, and $\Delta(x,y) = \infty$ otherwise. However, such a definition would be of no practical value since $\Delta(x,y)$ would be merely an indicator function. Thus, an effective limitation of our model is that it requires channels that are \emph{naturally} characterized by some discrepancy function. In particular, one should be able to interpret the maximum discrepancy $t$ as the level of ``degradedness'' of the channel.

On the other hand, the assumption $\Delta(x,y) \in \mathbb{N}$ imposes effectively no constraint. Since $|\calX \times \calY|$ is finite, given any ``naturally defined'' $\Delta' \colon \calX \times \calY \to \mathbb{R}$, one can always shift, scale and round the image of $\Delta'$ in order to produce some $\Delta \colon \calX \times \calY \to \mathbb{N}$ that induces the same fan-out sets as $\Delta'$ for all $t$.

\begin{example}\label{ex:t-error-channel}
  Let us use the above notation to define a $t$-error channel, i.e., a vector channel that introduces at most $t$ symbol errors (arbitrarily chosen by an adversary). Assume that the channel input and output alphabets are given by $\calX = \calY = \Fq^n$. It is easy to see that the channel can be characterized by a discrepancy function that counts the number of components in which an input vector $x$ and an output vector $y$ differ. More precisely, we have $\Delta(x,y) = \dH(x,y)$, where $\dH(\cdot,\cdot)$ denotes the \emph{Hamming distance} function.
\end{example}

A main feature of our proposed discrepancy characterization is to allow us to study a whole family of channels (with various levels of degradedness) under the same framework. For instance, we can use a single decoder for all channels in the same family. Define the \emph{minimum-discrepancy decoder} given by
\begin{equation}\label{eq:minimum-discrepancy-decoder}
  \hat{x} = \argmin_{x \in \calC}\, \Delta(x,y)
\end{equation}
where any ties in (\ref{eq:minimum-discrepancy-decoder}) are assumed to be broken arbitrarily. It is easy to see that a minimum-discrepancy decoder is infallible provided that the code is unambiguous. Thus, we can safely restrict attention to a minimum-discrepancy decoder, regardless of the maximum discrepancy $t$ in the channel.

\subsection{Correction Capability}

Given a fixed family of channels---specified by $\calX$, $\calY$ and $\Delta(\cdot,\cdot)$, and parameterized by a maximum discrepancy $t$---we wish to identify the largest (worst) channel parameter for which we can guarantee successful decoding. We say that a code is \emph{$t$-discrepancy-correcting} if it is unambiguous for a channel with maximum discrepancy $t$. The \emph{discrepancy-correction capability} of a code $\calC$ is the largest $t$ for which $\calC$ is $t$-discrepancy-correcting.

We start by giving a general characterization of the discrepancy-correction capability.
Let the function $\tau \colon \calX \times \calX \to \mathbb{N}$ be given by
\begin{equation}\label{eq:function-exact-capability}
  \tau(x,x') = \min_{y \in \calY}\, \max\{\Delta(x,y),\, \Delta(x',y)\} -1.
\end{equation}
We have the following result.

\begin{proposition}
  The discrepancy-correction capability of a code $\calC$ is given exactly by $\tau(\calC)$. In other words, $\calC$ is $t$-discrepancy-correcting if and only if $t \leq \tau(\calC)$.
\end{proposition}
\begin{proof}
 Suppose that the code is not $t$-discrepancy-correcting, i.e., that there exist some distinct $x,x' \in \calC$ and some $y \in \calY$ such that $\Delta(x,y) \leq t$ and $\Delta(x',y) \leq t$. Then $\tau(\calC) \leq \tau(x,x') \leq \max\{\Delta(x,y),\, \Delta(x',y)\} - 1 \leq t -1 < t$. In other words, $\tau(\calC) \geq t$ implies that the code is $t$-discrepancy-correcting.

 Conversely, suppose that $\tau(\calC) < t$, i.e., $\tau(\calC) \leq t-1$. Then there exist some distinct $x,x' \in \calC$ such that $\tau(x,x') \leq t-1$. This in turn implies that there exists some $y \in \calY$ such that $\max\{\Delta(x,y),\, \Delta(x',y)\} \leq t$. Since this implies that both $\Delta(x,y) \leq t$ and $\Delta(x',y) \leq t$, it follows that the code is not $t$-discrepancy-correcting.
\end{proof}

At this point, it is tempting to define a ``distance-like'' function given by $2(\tau(x,x') + 1)$, since this would enable us to immediately obtain results analogous to those of classical coding theory (such as the error correction capability being half the minimum distance of the code).
This approach has indeed been taken in previous works, such as \cite{Yang++2008:WeightProperties}.
Note, however, that the terminology ``distance'' suggests a geometrical interpretation, which is not immediately clear from (\ref{eq:function-exact-capability}). Moreover, the function (\ref{eq:function-exact-capability}) is not necessarily mathematically tractable.
It is the objective of this section to propose a ``distance'' function $\delta \colon \calX \times \calX \to \mathbb{N}$ that is motivated by geometrical considerations and is easier to handle analytically, yet is useful to characterize the correction capability of a code. In particular, we shall be able to obtain the same results as \cite{Yang++2008:WeightProperties} with much greater mathematical simplicity---which will later turn out to be instrumental for code design.

For $x,x' \in \calX$, define the \emph{$\Delta$-distance} between $x$ and $x'$ as
\begin{equation}\label{eq:Delta-distance-definition}
  \delta(x,x') \triangleq \min_{y \in \calY} \left\{\Delta(x,y) + \Delta(x',y)\right\}.
\end{equation}
The following interpretation holds. Consider the complete bipartite graph with vertex sets $\calX$ and $\calY$, and assume that each edge $(x,y) \in \calX \times \calY$ is labeled by a ``length'' $\Delta(x,y)$. Then $\delta(x,x')$ is the length of the shortest path between vertices $x,x' \in \calX$. Roughly speaking, $\delta(x,x')$ gives the minimum total effort that an adversary would have to spend (in independent channel realizations) in order to make $x$ and $x'$ both plausible explanations for some received output.

\begin{example}
  Let us compute the $\Delta$-distance for the channel of Example~\ref{ex:t-error-channel}. We have $\delta(x,x') = \min_{y} \left\{\dH(x,y) + \dH(x',y)\right\} \geq \dH(x,x')$, since the Hamming distance satisfies the triangle inequality. This bound is achievable by taking, for instance, $y=x'$. Thus, $\delta(x,x') = \dH(x,x')$, i.e., the $\Delta$-distance for this channel is given precisely by the Hamming distance.
\end{example}

The following result justifies our definition of the $\Delta$-distance.

\begin{proposition}\label{prop:correction-guarantee}
  For any code $\calC$, $\tau(\calC) \geq \lfloor (\delta(\calC) -1)/2 \rfloor$.
\end{proposition}
\begin{proof}
This follows from the fact that $\lfloor (a+b+1)/2 \rfloor \leq \max\{a,b\}$ for all $a,b \in \mathbb{Z}$.
\end{proof}

Proposition~\ref{prop:correction-guarantee} shows that $\delta(\calC)$ gives a lower bound on the correction capability of a code---therefore providing a correction guarantee. The converse result, however, is not necessarily true in general. Thus, up to this point, the proposed function is only partially useful: it is conceivable that the $\Delta$-distance might be too conservative and give a guaranteed correction capability that is lower than the actual one. Nevertheless, it is easier to deal with addition, as in (\ref{eq:Delta-distance-definition}), rather than maximization, as in (\ref{eq:function-exact-capability}).

A special case where the converse is true is for a family of channels whose discrepancy function satisfies the following condition:

\begin{definition}\label{def:magic-property}
  A discrepancy function $\Delta \colon \calX \times \calX \to \mathbb{N}$ is said to be \emph{\magic} if, for all $x,x' \in \calX$ and all $0\leq i \leq \delta(x,x')$, there exists some $y \in \calY$ such that $\Delta(x,y) = i$ and $\Delta(x',y) = \delta(x,x') - i$.
\end{definition}

\begin{theorem}\label{thm:magic-property-converse}
  Suppose that $\Delta(\cdot,\cdot)$ is \magic. For every code $\calC \subseteq \calX$, we have $\tau(\calC) = \lfloor (\delta(\calC)-1)/2 \rfloor$.
\end{theorem}
\begin{proof}
 We just need to show that $\lfloor (\delta(\calC)-1)/2 \rfloor \geq \tau(\calC)$.
 Take any $x,x' \in \calX$. Since $\Delta(\cdot,\cdot)$ is \magic, there exists some $y \in \calY$ such that $\Delta(x,y) + \Delta(x',y) = \delta(x,x')$ and either $\Delta(x,y) = \Delta(x',y)$ or $\Delta(x,y) = \Delta(x',y) - 1$. Thus, $\delta(x,x') \geq 2 \max\{\Delta(x,y),\Delta(x',y)\}-1$ and therefore $\lfloor (\delta(x,x')-1)/2 \rfloor \geq \tau(x,x')$.
\end{proof}

Theorem~\ref{thm:magic-property-converse} shows that, for certain families of channels, our proposed $\Delta$-distance achieves the goal of this section: it is a (seemingly) tractable function that precisely describes the correction capability of a code. In particular, the basic result of classical coding theory---that the Hamming distance precisely describes the error correction capability of a code---follows from the fact that the Hamming distance (as a discrepancy function) is \magic. As we shall see, much of our effort in the next sections reduces to showing that a specified discrepancy function is \magic.

Note that, for normal discrepancy functions, we actually have $\tau(x,x') = \lfloor (\delta(x,x')-1)/2 \rfloor$, so Theorem~\ref{thm:magic-property-converse} may also be regarded as providing an alternative (and more tractable) expression for $\tau(x,x')$.

\begin{example}
  To give a nontrivial example, let us consider a binary vector channel that introduces at most $\rho$ erasures (arbitrarily chosen by an adversary). The input alphabet is given by $\calX = \{0,1\}^n$, while the output alphabet is given by $\calX = \{0,1,\epsilon\}^n$, where $\epsilon$ denotes an erasure. We may define $\Delta(x,y) = \sum_{i=1}^n \Delta(x_i,y_i)$, where
\begin{equation}\nonumber
  \Delta(x_i,y_i) = \begin{cases}
                      0 & \text{if $y_i = x_i$} \\
                      1 & \text{if $y_i = \epsilon$} \\
                      \infty & \text{otherwise}
                     \end{cases}.
\end{equation}
The fan-out sets are then given by $\calY_x = \{y \in \calY \colon \Delta(x,y) \leq \rho\}$. In order to compute $\delta(x,x')$, observe the minimization in (\ref{eq:Delta-distance-definition}). It is easy to see that we should choose $y_i = x_i$ when $x_i = x_i'$, and $y_i = \epsilon$ when $x_i \neq x_i'$. It follows that $\delta(x,x') = 2\dH(x,x')$. Note that $\Delta(x,y)$ is \magic. It follows from Theorem~\ref{thm:magic-property-converse} that a code $\calC$ can correct all the $\rho$ erasures introduced by the channel if and only if $2\dH(\calC) > 2\rho$. This result precisely matches the well-known result of classical coding theory.
\end{example}

It is worth clarifying that, while we call $\delta(\cdot,\cdot)$ a ``distance,'' this function may not necessarily be a metric. While symmetry and non-negativity follow from the definition, a $\Delta$-distance may not always satisfy ``$\delta(x,y) = 0 \iff x=y$'' or the triangle inequality. Nevertheless, we keep the terminology for convenience.

Although this is not our main interest in this paper, it is worth pointing out that the framework of this section is also useful for obtaining results on error \emph{detection}. Namely, the $\Delta$-distance gives, in general, a lower bound on the discrepancy detection capability of a code under a bounded discrepancy-correcting decoder; when the discrepancy function is \magic, then the $\Delta$-distance precisely characterizes this detection capability (similarly as in classical coding theory). For more details on this topic, see Appendix~\ref{sec:detection-capability}.

\section{Coherent Network Coding}
\label{sec:coherent-network-coding}

\subsection{A Worst-Case Model and the Rank Metric}
\label{ssec:coherent-main}

The basic channel model for coherent network coding with adversarial errors is a matrix channel with input $X \in \Fq^{n \times m}$, output $Y \in \Fq^{N \times m}$, and channel law given by (\ref{eq:matrix-model}), where $A \in \Fq^{N \times n}$ is fixed and known to the receiver, and $Z \in \Fq^{t \times m}$ is arbitrarily chosen by an adversary. Here, we make the following additional assumptions:
\begin{itemize}
  \item The adversary has unlimited computational power and is omniscient; in particular, the adversary knows both $A$ and $X$;
  \item The matrix $D \in \Fq^{N \times t}$ is arbitrarily chosen by the adversary.
\end{itemize}
We also assume that $t < n$ (more precisely, we should assume $t < \rank A$); otherwise, the adversary may always choose $DZ = -AX$, leading to a trivial communications scenario.

The first assumption above allows us to use the approach of Section~\ref{sec:general-approach}. The second assumption may seem somewhat ``pessimistic,'' but it has the analytical advantage of eliminating from the problem any further dependence on the network code. (Recall that, in principle, $D$ would be determined by the network code and the choice of links in error.)

The power of the approach of Section~\ref{sec:general-approach} lies in the fact that the channel model defined above can be \emph{completely} described by the following discrepancy function
\begin{equation}\label{eq:coherent-discrepancy}
   \Delta_{A}(X,Y) \triangleq \min_{\substack{r \in \mathbb{N},\,D \in \Fq^{N \times r},\,Z \in \Fq^{r \times m}\colon \\ Y = AX + DZ}}\, r.
\end{equation}
The discrepancy $\Delta_{A}(X,Y)$ represents the minimum number of error packets that the adversary needs to inject in order to transform an input $X$ into an output $Y$, given that the transfer matrix is $A$. The subscript in $\Delta_A(X,Y)$ is to emphasize the dependence on~$A$.
For this discrepancy function, the minimum-discrepancy decoder becomes
\begin{equation}\label{eq:coherent-decoding-rule}
  \hat{X} = \argmin_{X \in \mathcal{C}}\, \Delta_{A}(X,Y).
\end{equation}
Similarly, the $\Delta$-distance induced by $\Delta_A(X,Y)$ is given by
\begin{equation}\label{eq:coherent-delta-between-codewords}
  \delta_{A}(X,X') \triangleq \min_{Y \in \Fq^{N \times m}} \left\{ \Delta_{A}(X,Y) + \Delta_{A}(X',Y)\right\}
\end{equation}
for $X,X' \in \Fq^{n \times m}$.

We now wish to find a simpler expression for $\Delta_A(X,Y)$ and $\delta_A(X,X')$, and show that $\Delta_A(X,Y)$ is \magic.

\begin{lemma}\label{lem:coherent-discrepancy-rank}
\begin{equation}\label{eq:coherent-discrepancy-rank}
  \Delta_{A}(X,Y) = \rank(Y - AX).
\end{equation}
\end{lemma}
\begin{proof}
Consider $\Delta_{A}(X,Y)$ as given by (\ref{eq:coherent-discrepancy}). For any feasible triple $(r,D,Z)$, we have $r \geq \rank Z \geq \rank DZ = \rank (Y-AX)$. This bound is achievable by setting $r = \rank (Y-AX)$ and letting $DZ$ be a full-rank decomposition of $Y-AX$.
\end{proof}

\begin{lemma}
  \begin{equation}\nonumber
    \delta_{A}(X,X') = \dr(AX,AX') = \rank A(X'-X).
  \end{equation}
\end{lemma}
\begin{proof}
From (\ref{eq:coherent-delta-between-codewords}) and Lemma~\ref{lem:coherent-discrepancy-rank}, we have $\delta_{A}(X,X') = \min_Y \left\{ \dr(Y,AX) + \dr(Y,AX') \right\}$. Since the rank metric satisfies the triangle inequality, we have $\dr(AX,Y) + \dr(AX',Y) \geq \dr(AX,AX')$. This lower bound can be achieved by choosing, e.g., $Y = AX$.
\end{proof}

Note that $\delta_A(\cdot,\cdot)$ is a metric if and only if $A$ has full column rank---in which case it is precisely the rank metric. (If $\rank A < n$, then there exist $X \neq X'$ such that $\delta_A(X,X') = 0$.)

\begin{theorem}\label{thm:coherent-magic}
  The discrepancy function $\Delta_A(\cdot,\cdot)$ is \magic.
\end{theorem}
\begin{proof}
Let $X,X' \in \Fq^{n \times m}$ and let $0 \leq i \leq d = \delta_A(X,X')$. Then $\rank A(X'-X) = d$. By performing a full-rank decomposition of $A(X'-X)$, we can always find two matrices $W$ and $W'$ such that $W + W' = A(X'-X)$, $\rank W = i$ and $\rank W' = d - i$. Taking $Y = AX + W = AX' - W'$, we have that $\Delta_{A}(X,Y) = i$ and $\Delta_{A}(X',Y) = d-i$.
\end{proof}

Note that, under the discrepancy $\Delta_A(X,Y)$, a $t$-discrepancy-correcting code is a code that can correct \emph{any} $t$ packet errors injected by the adversary. Using Theorem~\ref{thm:coherent-magic} and Theorem~\ref{thm:magic-property-converse}, we have the following result.

\begin{theorem}\label{thm:coherent-correction-capability}
  A code $\mathcal{C}$ is guaranteed to correct any $t$ packet errors if and only if $\delta_{A}(\mathcal{C}) > 2t$.
\end{theorem}

Theorem~\ref{thm:coherent-correction-capability} shows that $\delta_{A}(\mathcal{C})$ is indeed a fundamental parameter characterizing the error correction capability of a code in our model. Note that, if the condition of Theorem~\ref{thm:coherent-correction-capability} is violated, then there exists at least one codeword for which the adversary can certainly induce a decoding failure.

Note that the error correction capability of a code $\mathcal{C}$ is dependent on the network code through the matrix $A$.
Let $\rho = n - \rank A$ be the column-rank deficiency of $A$. Since $\delta_{A}(X,X') = \rank A(X'-X)$, it follows from (\ref{eq:bound-rank-product}) that
\begin{equation}\nonumber
  \dr(X,X') - \rho \leq \delta_{A}(X,X') \leq \dr(X,X')
\end{equation} and
\begin{equation}\label{eq:coherent-delta-rank-distance}
  \dr(\mathcal{C}) - \rho \leq \delta_{A}(\mathcal{C}) \leq \dr(\mathcal{C}).
\end{equation}
Thus, the error correction capability of a code is strongly tied to its minimum rank distance; in particular, $\delta_A(\calC) = \dr(\calC)$ if $\rho = 0$.
While the lower bound $\delta_{A}(\calC) \geq \dr(\calC) - \rho$ may not be tight in general, we should expect it to be tight when $\calC$ is sufficiently large. This is indeed the case for MRD codes, as discussed in Section~\ref{ssec:coherent-optimality}. Thus, a rank deficiency of $A$ will typically reduce the error correction capability of a code.

Taking into account the worst case, we can use Theorem~\ref{thm:coherent-correction-capability} to give a correction guarantee in terms of the minimum rank distance of the code.

\begin{proposition}\label{prop:coherent-guarantee-rank-metric}
  A code $\mathcal{C}$ is guaranteed to correct $t$ packet errors, under rank deficiency $\rho$,
 if $\dr(\mathcal{C}) > 2t + \rho$.
\end{proposition}

Note that the guarantee of Proposition~\ref{prop:coherent-guarantee-rank-metric} depends only on $\rho$ and $t$; in particular, it is independent of the network code or the specific transfer matrix $A$.

\subsection{Reinterpreting the Model of Yeung et al.}
\label{ssec:coherent-comparison}

In this subsection, we investigate the model for coherent network coding studied by Yeung et al.\ in \cite{Cai.Yeung2002,Yeung.Cai2006,Cai.Yeung2006,Yang.Yeung2007}, which is similar to the one considered in the previous subsection. The model is that of a matrix channel with input $X \in \Fq^{n \times m}$, output $Y \in \Fq^{N \times m}$, and channel law given by
\begin{equation}\label{eq:channel-model-full}
  Y = AX + FE
\end{equation}
where $A \in \Fq^{N \times n}$ and $F \in \Fq^{N \times |\mathcal{E}|}$ are fixed and known to the receiver, and $E \in \Fq^{|\mathcal{E}| \times m}$ is arbitrarily chosen by an adversary provided $\wt(E) \leq t$. (Recall that $|\calE|$ is the number of edges in the network.) In addition, the adversary has unlimited computational power and is omniscient, knowing, in particular, $A$, $F$ and $X$.


We now show that some of the concepts defined in \cite{Yang.Yeung2007}, such as ``network Hamming distance,'' can be reinterpreted in the framework of Section~\ref{sec:general-approach}. As a consequence, we can easily recover the results of \cite{Yang.Yeung2007} on error correction and detection guarantees.

First, note that the current model can be completely described by the following discrepancy function
\begin{equation}\label{eq:Yeung-discrepancy}
   \Delta_{A,F}(X,Y) \triangleq \min_{\substack{E \in \Fq^{|\calE| \times m}\colon \\ Y = AX + FE}}\, \wt(E).
\end{equation}
The $\Delta$-distance induced by this discrepancy function is given by
\begin{align}
\delta_{A,F}(X_1,X_2)
&\triangleq \min_{Y}\, \Delta_{A,F}(X_1,Y) + \Delta_{A,F}(X_2,Y) \nonumber \\
&= \min_{\substack{Y,E_1,E_2:\\ Y=AX_1 + FE_1\\ Y=AX_2 + FE_2}} \left\{ \wt(E_1) + \wt(E_2) \right\} \nonumber \\
&= \min_{\substack{E_1,E_2\colon\\ A(X_2 - X_1) = F(E_1-E_2)}} \left\{ \wt(E_1) + \wt(E_2) \right\} \nonumber \\
&= \min_{\substack{E\colon\\ A(X_2 - X_1) = FE}}\, \wt(E) \nonumber \label{eq:proof-NHD-1}
\end{align}
where the last equality follows from the fact that $\wt(E_1 - E_2) \leq \wt(E_1) + \wt(E_2)$, achievable if $E_1 = 0$.

Let us now examine some of the concepts defined in \cite{{Yang.Yeung2007}}. For a specific sink node, the decoder proposed in \cite[Eq. (2)]{Yang.Yeung2007} has the form
\begin{equation}\nonumber
  \hat{X} = \argmin_{X \in \mathcal{C}}\, \Psi_{A,F}(X,Y).
\end{equation}
The definition of the objective function $\Psi_{A,F}(X,Y)$ requires several other definitions presented in \cite{Yang.Yeung2007}. Specifically, $\Psi_{A,F}(X,Y) \triangleq D^{rec}(AX,Y)$, where $D^{rec}(Y_1, Y_2) \triangleq W^{rec}(Y_2 - Y_1)$, $W^{rec}(Y) \triangleq \min_{E \in \Upsilon(Y)}\, \wt(E)$, and $\Upsilon(Y) \triangleq \{E\colon Y=FE\}$.
Substituting all these values into $\Psi_{A,F}(X,Y)$, we obtain
\begin{align}
\Psi_{A,F}(X,Y)
&= D^{rec}(AX,Y) \nonumber \\
&= W^{rec}(Y-AX) \nonumber \\
&= \min_{E \in \Upsilon(Y-AX)}\, \wt(E) \nonumber \\
&= \min_{E\colon Y-AX = FE}\, \wt(E) \nonumber \\
&= \Delta_{A,F}(X,Y). \nonumber
\end{align}
Thus, the decoder in \cite{Yang.Yeung2007} is precisely a minimum-discrepancy decoder.


In \cite{Yang.Yeung2007}, the ``network Hamming distance'' between two messages $X_1$ and $X_2$ is defined as $D^{msg}(X_1,X_2) \triangleq W^{msg}(X_2 - X_1)$, where $W^{msg}(X) \triangleq W^{rec}(AX)$. Again, simply substituting the corresponding definitions yields
\begin{align}
D^{msg}(X_1,X_2)
&= W^{msg}(X_2 - X_1) \nonumber \\
&= W^{rec}(A(X_2 - X_1)) \nonumber \\
&= \min_{E \in \Upsilon(A(X_2 - X_1))}\, \wt(E) \nonumber \\
&= \min_{E \colon A(X_2 - X_1) = FE}\, \wt(E) \nonumber \\
&= \delta_{A,F}(X_1,X_2). \nonumber
\end{align}
Thus, the ``network Hamming distance'' is precisely the $\Delta$-distance induced by the discrepancy function $\Delta_{A,F}(X,Y)$. Finally, the ``unicast minimum distance'' of a network code with message set $\mathcal{C}$ \cite{Yang.Yeung2007} is precisely $\delta_{A,F}(\calC)$.

Let us return to the problem of characterizing the correction capability of a code.

\begin{proposition}
 The discrepancy function $\Delta_{A,F}(\cdot,\cdot)$ is \magic.
\end{proposition}
\begin{proof}
 Let $X_1,X_2 \in \Fq^{n \times m}$ and let $0 \leq i \leq d = \delta_{A,F}(X_1,X_2)$. Let $E \in \Fq^{|\calE| \times m}$ be a solution to the minimization in (\ref{eq:Yeung-discrepancy}). Then $A(X_2 - X_1) = FE$ and $\wt(E) = d$. By partitioning $E$, we can always find two matrices $E_1$ and $E_2'$ such that $E_1 + E_2' = E$, $\rank E_1 = i$ and $\rank E_2' = d - i$. Taking $Y = AX_1 + FE_1 = AX_2 - FE_2$, we have that $\Delta_{A,F}(X_1,Y) \leq i$ and $\Delta_{A,F}(X_2,Y) \leq d-i$. Since $d \leq \Delta_{A,F}(X_1,Y)  + \Delta_{A,F}(X_2,Y)$, it follows that $\Delta_{A,F}(X_1,Y) = i$ and $\Delta_{A,F}(X_2,Y) = d-i$.
\end{proof}

It follows that a code $\calC$ is guaranteed to correct any $t$ packet errors if and only if $\delta_{A,F}(\calC) > 2t$.
Thus, we recover theorems 2 and 3 in \cite{Yang.Yeung2007} (for error detection, see Appendix~\ref{sec:detection-capability}). The analogous results for the multicast case can be obtained in a straightforward manner.

We now wish to compare the parameters devised in this subsection with those of Section~\ref{ssec:coherent-main}. From the descriptions of (\ref{eq:matrix-model}) and (\ref{eq:channel-model-full}), it is intuitive that the model of this subsection should be equivalent to that of the previous subsection if the matrix $F$, rather than fixed and known to the receiver, is arbitrarily and secretly chosen by the adversary. A formal proof of this fact is given in the following proposition.

\begin{proposition}\label{prop:coherent-Yeung-comparison}
\begin{align}
\Delta_{A}(X,Y) &= \min_{F \in \Fq^{N \times |\calE|}}\, \Delta_{A,F}(X,Y) \nonumber \\
\delta_{A}(X,X') &= \min_{F \in \Fq^{N \times |\calE|}}\,\delta_{A,F}(X,X') \nonumber \\
\delta_{A}(\mathcal{C}) &= \min_{F \in \Fq^{N \times |\calE|}}\,\delta_{A,F}(\mathcal{C}). \nonumber
\end{align}
\end{proposition}
\begin{proof}
Consider the minimization
\begin{equation}\nonumber
  \min_{F \in \Fq^{N \times |\calE|}}\, \Delta_{A,F}(X,Y) = \min_{\substack{F \in \Fq^{N \times |\calE|},\, E \in \Fq^{|\calE| \times m}\colon \\ Y = AX + FE}}\, \wt(E).
\end{equation}
For any feasible $(F,E)$, we have $\wt(E) \geq \rank E \geq \rank FE = \rank (Y-AX)$. This lower bound can be achieved by taking
\begin{equation}\nonumber
F = \mat{F' & 0} \quad\text{and}\quad E = \mat{E' \\ 0}
\end{equation}
where $F'E'$ is a full-rank decomposition of $Y-AX$. This proves the first statement. The second statement follows from the first by noticing that $\delta_A(X,X') = \Delta_A(X,AX')$ and $\delta_{A,F}(X,X') = \Delta_{A,F}(X,AX')$. The third statement is immediate.
\end{proof}

Proposition~\ref{prop:coherent-Yeung-comparison} shows that the model of Section~\ref{ssec:coherent-main} is indeed more pessimistic, as the adversary has additional power to choose the worst possible $F$. It follows that any code that is $t$-error-correcting for that model must also be $t$-error-correcting for the model of Yeung et al.


\subsection{Optimality of MRD Codes}
\label{ssec:coherent-optimality}


Let us now evaluate the performance of an MRD code under the models of the two previous subsections.

The Singleton bound of \cite{Yeung.Cai2006} (see also \cite{Yang.Yeung2007:RefinedBounds}) states that
\begin{equation}\label{eq:Singleton-Yeung}
  |\mathcal{C}| \leq Q^{n-\rho - \delta_{A,F}(\mathcal{C}) + 1}
\end{equation}
where $Q$ is the size of the alphabet\footnote{This alphabet is usually assumed a finite field, but, for the Singleton bound of \cite{Yeung.Cai2006}, it is sufficient to assume an abelian group, e.g., a vector space over $\Fq$.} from which packets are drawn. Note that $Q = q^m$ in our setting, since each packet consists of $m$ symbols from $\Fq$. Using Proposition~\ref{prop:coherent-Yeung-comparison}, we can also obtain
\begin{equation}\label{eq:Singleton-coherent}
  |\mathcal{C}| \leq Q^{n-\rho - \delta_{A}(\mathcal{C}) + 1} = q^{m(n-\rho - \delta_{A}(\mathcal{C}) + 1)}.
\end{equation}

On the other hand, the size of an MRD code, for $m \geq n$, is given by
\begin{align}
|\mathcal{C}|
&= q^{m(n-\dr(\mathcal{C})+1)} \label{eq:coherent-MRD-Singleton} \\
&\geq q^{m(n-\rho-\delta_{A}(\mathcal{C})+1)} \label{eq:coherent-MRD-proof-1} \\
&\geq q^{m(n-\rho-\delta_{A,F}(\mathcal{C})+1)} \nonumber
\end{align}
where (\ref{eq:coherent-MRD-proof-1}) follows from (\ref{eq:coherent-delta-rank-distance}). Since $Q = q^m$, both (\ref{eq:Singleton-Yeung}) and (\ref{eq:Singleton-coherent}) are achieved in this case. Thus, we have the following result.

\begin{theorem}\label{thm:coherent-optimality-MRD}
  When $m \geq n$, an MRD code $\calC \subseteq \Fq^{n \times m}$ achieves maximum cardinality with respect to both $\delta_{A}$ and $\delta_{A,F}$.
\end{theorem}

Theorem~\ref{thm:coherent-optimality-MRD} shows that, if an alphabet of size $Q = q^m \geq q^n$ is allowed (i.e., a packet size of at least $n \log_2 q$ bits), then MRD codes turn out to be optimal under both models of sections \ref{ssec:coherent-main} and \ref{ssec:coherent-comparison}.

\begin{remark}
  It is straightforward to extend the results of Section~\ref{ssec:coherent-main} for the case of multiple heterogeneous receivers, where each receiver $u$ experiences a rank deficiency $\rho^{(u)}$. In this case, it can be shown that an MRD code with $m \geq n$ achieves the refined Singleton bound of \cite{Yang.Yeung2007:RefinedBounds}.
\end{remark}

Note that, due to (\ref{eq:Singleton-coherent}), (\ref{eq:coherent-MRD-Singleton}) and (\ref{eq:coherent-MRD-proof-1}), it follows that $\delta_A(\calC) = \dr(\calC) - \rho$ for an MRD code with $m \geq n$. Thus, in this case, we can restate Theorem~\ref{thm:coherent-correction-capability} in terms of the minimum rank distance of the code.

\begin{theorem}\label{thm:coherent-complete-capability}
  An MRD code $\mathcal{C} \subseteq \Fq^{n \times m}$ with $m \geq n$ is guaranteed to correct $t$ packet errors, under rank deficiency $\rho$,
 if and only if $\dr(\mathcal{C}) > 2t + \rho$.
\end{theorem}

Observe that Theorem~\ref{thm:coherent-complete-capability} holds regardless of the specific transfer matrix $A$, depending only on its column-rank deficiency $\rho$.

The results of this section imply that, when designing a linear network code, we may focus solely on the objective of making the network code feasible, i.e., maximizing $\rank A$. 
If an error correction guarantee is desired, then an outer code can be applied end-to-end without requiring any modifications on (or even knowledge of) the underlying network code. The design of the outer code is essentially trivial, as any MRD code can be used, with the only requirement that the number of $\Fq$-symbols per packet, $m$, is at least $n$.

\begin{remark}
  Consider the decoding rule (\ref{eq:coherent-decoding-rule}). The fact that (\ref{eq:coherent-decoding-rule}) together with (\ref{eq:coherent-discrepancy-rank}) is equivalent to \cite[Eq. (20)]{Silva++2008} implies that the decoding problem can be solved by exactly the same rank-metric techniques proposed in \cite{Silva++2008}. In particular, for certain MRD codes with $m \geq n$ and minimum rank distance $d$, there exist efficient encoding and decoding algorithms both requiring $O(dn^2m)$ operations in $\Fq$ per codeword. For more details, see \cite{Silva2009(PhD)}.
\end{remark}

\section{Noncoherent Network Coding}
\label{sec:noncoherent-network-coding}

\subsection{A Worst-Case Model and the Injection Metric}
\label{ssec:noncoherent-main}

Our model for noncoherent network coding with adversarial errors differs from its coherent counterpart of Section~\ref{ssec:coherent-main} only with respect to the transfer matrix $A$. Namely, the matrix $A$ is unknown to the receiver and is freely chosen by the adversary while respecting the constraint $\rank A \geq n - \rho$. The parameter $\rho$, the maximum column rank deficiency of $A$, is a parameter of the system that is known to all. Note that, as discussed above for the matrix $D$, the assumption that $A$ is chosen by the adversary is what provides the conservative (worst-case) nature of the model. The constraint on the rank of $A$ is required for a meaningful coding problem; otherwise, the adversary could prevent communication by simply choosing $A=0$.

As before, we assume a minimum-discrepancy decoder
\begin{equation}\label{eq:noncoherent-decoding-rule}
  \hat{X} = \argmin_{X \in \mathcal{C}}\, \Delta_{\rho}(X,Y)
\end{equation}
with discrepancy function given by
\begin{align}
\Delta_{\rho}(X,Y) &\triangleq \min_{\substack{A \in \Fq^{N \times n},r \in \mathbb{N}, D \in \Fq^{N \times r},Z \in \Fq^{r \times m}\colon \\ Y = AX + DZ \\ \rank A \geq n-\rho}}\, r \label{eq:noncoherent-discrepancy} \\
&= \min_{\substack{A \in \Fq^{N \times n}:\\ \rank A \geq n - \rho}}\, \Delta_{A}(X,Y). \nonumber
\end{align}
Again, $\Delta_{\rho}(X,Y)$ represents the minimum number of error packets needed to produce an output $Y$ given an input $X$ under the current adversarial model. The subscript is to emphasize that $\Delta_{\rho}(X,Y)$ is still a function of $\rho$.

The $\Delta$-distance induced by $\Delta_\rho(X,Y)$ is defined below. For $X,X' \in \Fq^{n \times m}$, let
\begin{equation}\label{eq:noncoherent-delta-between-codewords}
  \delta_{\rho}(X,X') \triangleq \min_{Y \in \Fq^{N \times m}} \left\{ \Delta_{\rho}(X,Y) + \Delta_{\rho}(X',Y) \right\}.
\end{equation}

We now prove that $\Delta_\rho(X,Y)$ is {\magic} and therefore $\delta_\rho(\calC)$ characterizes the correction capability of a code.

First, observe that, using Lemma~\ref{lem:coherent-discrepancy-rank}, we may rewrite $\Delta_{\rho}(X,Y)$ as
\begin{equation}\label{eq:noncoherent-discrepancy-rank}
  \Delta_{\rho}(X,Y) = \min_{\substack{A \in \Fq^{N \times n}:\\ \rank A \geq n - \rho}} \rank(Y - AX).
\end{equation}
Also, note that
  \begin{align}
  \delta_{\rho}(X,X')
  &= \min_{Y} \left\{ \Delta_{\rho}(X,Y) + \Delta_{\rho}(X',Y) \right\} \nonumber \\
  &= \min_{\substack{A,A' \in \Fq^{N \times n}:\\ \rank A \geq n - \rho\\ \rank A' \geq n - \rho}} \{ \min_{Y}\, \rank(Y - AX) \nonumber \\[-5ex]
  &\hspace{8.5em} {} + \rank(Y - A'X') \} \nonumber \\[2ex]
  &= \min_{\substack{A,A' \in \Fq^{N \times n}:\\ \rank A \geq n - \rho\\ \rank A' \geq n - \rho}} \rank(A'X' - AX) \label{eq:noncoherent-delta-expression-rank}
  \end{align}
where the last equality follows from the fact that $\dr(AX,Y) + \dr(A'X',Y) \geq \dr(AX, A'X')$, achievable by choosing, e.g., $Y = AX$.

\begin{theorem}\label{thm:noncoherent-magic}
  The discrepancy function $\Delta_\rho(\cdot,\cdot)$ is \magic.
\end{theorem}
\begin{proof}
Let $X,X' \in \Fq^{n \times m}$ and let $0 \leq i \leq d = \delta_\rho(X,X')$. Let $A,A' \in \Fq^{N \times n}$ be a solution to the minimization in (\ref{eq:noncoherent-delta-expression-rank}). Then $\rank (A'X' - AX) = d$. By performing a full-rank decomposition of $A'X'-AX$, we can always find two matrices $W$ and $W'$ such that $W + W' = A'X'-AX$, $\rank W = i$ and $\rank W' = d - i$. Taking $Y = AX + W = AX' - W'$, we have that $\Delta_\rho(X,Y) \leq i$ and $\Delta_\rho(X',Y) \leq d-i$.
Since $d \leq \Delta_\rho(X,Y) + \Delta_\rho(X',Y)$, it follows that $\Delta_\rho(X,Y) = i$ and $\Delta_\rho(X',Y) = d-i$.
\end{proof}

As a consequence of Theorem~\ref{thm:noncoherent-magic}, we have the following result.

\begin{theorem}\label{thm:noncoherent-capability}
  A code $\mathcal{C}$ is guaranteed to correct any $t$ packet errors if and only if $\delta_\rho(\mathcal{C}) > 2t$.
\end{theorem}

Similarly as in Section~\ref{ssec:coherent-main}, Theorem~\ref{thm:noncoherent-capability} shows that $\delta_{\rho}(\mathcal{C})$ is a fundamental parameter characterizing the error correction capability of a code in the current model. In contrast to Section~\ref{ssec:coherent-main}, however, the expression for $\Delta_\rho(X,Y)$ (and, consequently, $\delta_\rho(X,X')$) does not seem mathematically appealing since it involves a minimization. We now proceed to finding simpler expressions for $\Delta_\rho(X,Y)$ and $\delta_\rho(X,X')$.

The minimization in (\ref{eq:noncoherent-discrepancy-rank}) is a special case of a more general expression, which we give as follows.
For $X \in \Fq^{n \times m}$, $Y \in \Fq^{N \times m}$ and $L \geq \max\{n-\rho,\,N-\sigma\}$, let
\begin{equation}\nonumber
  \Delta_{\rho,\sigma,L}(X,Y) \triangleq \min_{\substack{A \in \Fq^{L \times n}, B \in \Fq^{L \times N}\colon \\ \rank A \geq n-\rho \\ \rank B \geq N - \sigma}}\, \rank(BY - AX).
\end{equation}

The quantity defined above is computed in the following lemma.

\begin{lemma}\label{lem:min-rank-quantity}
\begin{equation}\nonumber
  \Delta_{\rho,\sigma,L}(X,Y) = \Big[\max\{\rank X - \rho,\, \rank Y - \sigma\} - \dim (\linspan{X} \cap \linspan{Y})\Big]^+.
\end{equation}
\end{lemma}
\begin{proof}
See Appendix~\ref{sec:proof-lemma}.
\end{proof}

Note that $\Delta_{\rho,\sigma,L}(X,Y)$ is independent of $L$, for all valid $L$. Thus, we may drop the subscript and write simply $\Delta_{\rho,\sigma}(X,Y) \triangleq \Delta_{\rho,\sigma,L}(X,Y)$.

We can now provide a simpler expression for $\Delta_{\rho}(X,Y)$.

\begin{theorem}\label{thm:noncoherent-discrepancy-expression}
  \begin{equation}\nonumber
    \Delta_{\rho}(X,Y) = \max\{\rank X - \rho,\, \rank Y\} - \dim (\linspan{X} \cap \linspan{Y}).
 \end{equation}
\end{theorem}
\begin{proof}
  This follows immediately from Lemma~\ref{lem:min-rank-quantity} by noticing that $\Delta_{\rho}(X,Y) = \Delta_{\rho,0}(X,Y)$.
\end{proof}

From Theorem~\ref{thm:noncoherent-discrepancy-expression}, we observe that $\Delta_{\rho}(X,Y)$ depends on the matrices $X$ and $Y$ only through their row spaces, i.e., only the transmitted and received row spaces have a role in the decoding. Put another way, we may say that the channel really accepts an input subspace $\linspan{X}$ and delivers an output subspace $\linspan{Y}$. Thus, all the communication is made via subspace selection. This observation provides a fundamental justification for the approach of \cite{Kotter.Kschischang2008}.

At this point, it is useful to introduce the following definition.

\begin{definition}\label{def:modif-subspace-distance}
  The \emph{injection distance} between subspaces $\mathcal{U}$ and $\mathcal{V}$ in $\calP_q(m)$ is defined as
\begin{align}
  \di(\mathcal{U},\mathcal{V})
  &\triangleq \max\{\dim \mathcal{U},\, \dim \mathcal{V}\} - \dim(\mathcal{U} \cap \mathcal{V})  \label{eq:injection-distance-definition} \\
  &= \dim(\mathcal{U} + \mathcal{V}) - \min\{\dim \mathcal{U},\, \dim \mathcal{V}\}. \nonumber
\end{align}
\end{definition}

The injection distance can be interpreted as measuring the number of error packets that an adversary needs to inject in order to transform an input subspace $\linspan{X}$ into an output subspace $\linspan{Y}$. This can be clearly seen from the fact that $\di(\linspan{X},\linspan{Y}) = \Delta_{0}(X,Y)$. Thus, the injection distance is essentially equal to the discrepancy $\Delta_{\rho}(X,Y)$ when the channel is influenced only by the adversary, i.e., when the non-adversarial aspect of the channel (the column-rank deficiency of $A$) is removed from the problem. Note that, in this case
, the decoder (\ref{eq:noncoherent-decoding-rule}) becomes precisely a minimum-injection-distance decoder.

\begin{proposition}\label{prop:injection-metric}
  The injection distance is a metric.
\end{proposition}

We delay the proof of Proposition~\ref{prop:injection-metric} until Section~\ref{ssec:noncoherent-comparison}.

We can now use the definition of the injection distance to simplify the expression for the $\Delta$-distance.

\begin{proposition}\label{prop:noncoherent-delta-expression}
  \begin{equation}\nonumber
  \delta_{\rho}(X,X') = [\di(\linspan{X},\linspan{X'}) - \rho]^+.
  \end{equation}
\end{proposition}
\begin{proof}
This follows immediately after realizing that $\delta_{\rho}(X,X') = \Delta_{\rho,\rho}(X,X')$.
\end{proof}

From Proposition~\ref{prop:noncoherent-delta-expression}, it is clear that $\delta_\rho(\cdot,\cdot)$ is a metric if and only if $\rho = 0$ (in which case it is precisely the injection metric). If $\rho > 0$, then $\delta_\rho(\cdot,\cdot)$ does not satisfy the triangle inequality.

It is worth noticing that $\delta_{\rho}(X,X') = 0$ for any two matrices $X$ and $X'$ that share the same row space. Thus, any reasonable code $\mathcal{C}$ should avoid this situation.

For $\mathcal{C} \subseteq \Fq^{n \times m}$, let
\begin{equation}\nonumber
  \linspan{\mathcal{C}} = \{\linspan{X} \colon X \in \mathcal{C}\}
\end{equation}
be the subspace code (i.e., a collection of subspaces) consisting of the row spaces of all matrices in $\calC$.
The following corollary of Proposition~\ref{prop:noncoherent-delta-expression} is immediate.

\begin{corollary}\label{cor:noncoherent-delta-min-expression}
  Suppose $\mathcal{C}$ is such that $|\mathcal{C}| = |\linspan{\mathcal{C}}|$, i.e., no two codewords of $\mathcal{C}$ have the same row space. Then
\begin{equation}\nonumber
  \delta_{\rho}(\mathcal{C}) = [\di(\linspan{\mathcal{C}}) - \rho]^+.
\end{equation}
\end{corollary}

Using Corollary~\ref{cor:noncoherent-delta-min-expression}, we can restate Theorem~\ref{thm:noncoherent-capability}
more simply in terms of the injection distance.

\begin{theorem}\label{thm:noncoherent-complete-capability}
  A code $\mathcal{C}$ is guaranteed to correct $t$ packet errors, under rank deficiency $\rho$, if and only if $\di(\linspan{\mathcal{C}}) > 2t + \rho$.
\end{theorem}

Note that, due to equality in Corollary~\ref{cor:noncoherent-delta-min-expression}, a converse is indeed possible in Theorem~\ref{thm:noncoherent-complete-capability} (contrast with Proposition~\ref{prop:coherent-guarantee-rank-metric} for the coherent case).

Theorem~\ref{thm:noncoherent-complete-capability} shows that $\di(\linspan{\calC})$ is a fundamental parameter characterizing the \emph{complete} correction capability (i.e., error correction capability and ``rank-deficiency correction'' capability) of a code in our noncoherent model. Put another way, we may say that a code $\mathcal{C}$ is good for the model of this subsection if and only if its subspace version $\linspan{\mathcal{C}}$ is a good code in the injection metric.

\subsection{Comparison with the Metric of K\"otter and Kschischang}
\label{ssec:noncoherent-comparison}

Let $\Omega \subseteq \calP_q(m)$ be a subspace code whose elements have maximum dimension $n$. In \cite{Kotter.Kschischang2008}, the network is modeled as an operator channel that takes in a subspace $\mathcal{V} \in \calP_q(m)$ and puts out a possibly different subspace $\mathcal{U} \in \calP_q(m)$. The kind of disturbance that the channel applies to $\mathcal{V}$ is captured by the notions of ``insertions'' and ``deletions'' of dimensions (represented mathematically using operators), and the degree of such a dissimilarity is captured by the subspace distance
\begin{align}
\ds(\mathcal{V},\mathcal{U})
&\triangleq \dim(\calV + \calU) - \dim(\calV \cap \calU) \nonumber \\
&= \dim \calV + \dim \calU - 2 \dim(\calV \cap \calU) \label{eq:subspace-distance-definition} \\
&= \dim(\calV + \calU) - \dim \calV - \dim \calU. \nonumber
\end{align}
The transmitter selects some $\mathcal{V} \in \Omega$ and transmits $\mathcal{V}$ over the channel. The receiver receives some subspace $\mathcal{U}$ and, using a \emph{minimum subspace distance decoder}, decides that the subspace $\hat{\mathcal{V}} \subseteq \Omega$ was sent, where
\begin{equation}\label{eq:subspace-decoding-rule}
  \hat{\mathcal{V}} = \argmin_{\mathcal{V} \in \Omega}\, \ds(\mathcal{V},\mathcal{U}).
\end{equation}
This decoder is guaranteed to correct all disturbances applied by the channel if $\ds(\mathcal{V},\mathcal{U}) < \ds(\Omega)/2$, where $\ds(\Omega)$ is the minimum subspace distance between all pairs of distinct codewords of $\Omega$.

First, let us point out that this setup is indeed the same as that of Section~\ref{ssec:noncoherent-main} if we set $\mathcal{V} = \linspan{X}$, $\mathcal{U} = \linspan{Y}$ and $\Omega = \linspan{\mathcal{C}}$, where $\mathcal{C}$ is such that $|\mathcal{C}| = |\linspan{\mathcal{C}}|$. Also, any disturbance applied by an operator channel can be realized by a matrix model, and vice-versa. Thus, the difference between the approach of this section and that of \cite{Kotter.Kschischang2008} lies in the choice of the decoder.

Indeed, by using Theorem~\ref{thm:noncoherent-discrepancy-expression} and the definition of subspace distance, we get the following relationship:

\begin{proposition}\label{prop:Delta-and-subspace-distance}
  \begin{equation}\nonumber
    \Delta_{\rho}(X,Y) = \frac{1}{2} \ds(\linspan{X},\linspan{Y}) - \frac{1}{2} \rho + \frac{1}{2} | \rank X - \rank Y - \rho |.
  \end{equation}
\end{proposition}

Thus, we can see that when the matrices in $\mathcal{C}$ do not all have the same rank (i.e., $\Omega$ is a \emph{non-constant-dimension code}), then the decoding rules (\ref{eq:noncoherent-decoding-rule}) and (\ref{eq:subspace-decoding-rule}) may produce different decisions.

Using $\rho=0$ in the above proposition (or simply using (\ref{eq:injection-distance-definition}) and (\ref{eq:subspace-distance-definition})) gives us another formula for the injection distance:
\begin{equation}\label{eq:injection-and-subspace-distance}
  \di(\calV,\calU) = \frac{1}{2} \ds(\mathcal{V},\mathcal{U}) + \frac{1}{2} | \dim \mathcal{V} - \dim \mathcal{U} |.
\end{equation}

We can now prove a result that was postponed in the previous section.

\begin{theorem}
  The injection distance is a metric.
\end{theorem}
\begin{proof}
  Since $\ds(\cdot,\cdot)$ is a metric on $\calP_q(m)$ and $|\cdot|$ is a norm on $\mathbb{R}$, it follows from (\ref{eq:injection-and-subspace-distance}) that $\di(\cdot,\cdot)$ is also a metric on $\calP_q(m)$.
\end{proof}

We now examine in more detail an example situation where the minimum-subspace-distance decoder and the minimum-discrepancy decoder produce different decisions.

\begin{example}\label{ex:decoding-comparison}
For simplicity, assume $\rho = 0$. Consider a subspace code that contains two codewords
$\mathcal{V}_1 = \linspan{X_1}$ and $\mathcal{V}_2 = \linspan{X_2}$
such that $\gamma \triangleq \dim \mathcal{V}_2 - \dim \mathcal{V}_1$ satisfies $d/3 < \gamma < d/2$, where $d \triangleq \ds(\mathcal{V}_1,\mathcal{V}_2)$.

Suppose the received subspace $\mathcal{U} = \linspan{Y}$ is such that
$\mathcal{V}_1 \subseteq \mathcal{U} \subseteq \mathcal{V}_1 + \mathcal{V}_2$ and $\dim \mathcal{U} = \dim \mathcal{V}_1 + \gamma = \dim \mathcal{V}_2$,
as illustrated in Fig.~\ref{fig:subspaces}.
Then $\ds(\mathcal{V}_1,\mathcal{U}) = \gamma$ and $\ds(\mathcal{V}_2,\mathcal{U}) = d- \gamma$, while Proposition~(\ref{prop:Delta-and-subspace-distance}) gives $\Delta_{\rho}(X_1,Y) = \gamma$ and $\Delta_{\rho}(X_2,Y) = (d-\gamma)/2 \triangleq \epsilon$. Since, by assumption, $d-\gamma > \gamma$ and $\epsilon < \gamma$, it follows that $\ds(\mathcal{V}_1,\mathcal{U}) < \ds(\mathcal{V}_2,\mathcal{U})$ but $\Delta_{\rho}(X_1,Y) > \Delta_{\rho}(X_2,Y)$, i.e., the decoders (\ref{eq:subspace-decoding-rule}) and (\ref{eq:noncoherent-decoding-rule}) will produce different decisions.

This situation can be intuitively explained as follows. The decoder (\ref{eq:subspace-decoding-rule}) favors the subspace $\mathcal{V}_1$, which is closer in subspace distance to $\mathcal{U}$ than $\mathcal{V}_2$. However, since $\mathcal{V}_1$ is low-dimensional, $\mathcal{U}$ can only be produced from $\mathcal{V}_1$ by the \emph{insertion} of $\gamma$ dimensions. The decoder (\ref{eq:noncoherent-decoding-rule}), on the other hand, favors $\mathcal{V}_2$, which, although farther in subspace distance, can produce $\mathcal{U}$ after the \emph{replacement} of $\epsilon < \gamma$ dimensions. Since one packet error must occur for each inserted or replaced dimension, we conclude that the decoder (\ref{eq:noncoherent-decoding-rule}) finds the solution that minimizes the number of packet errors observed.
\end{example}

\begin{figure}
  \centering
  \includegraphics[scale=1.0]{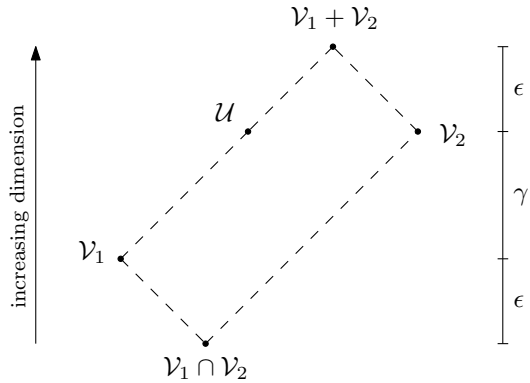}\\
  \caption{Lattice of subspaces in Example~\ref{ex:decoding-comparison}. Two spaces are joined with a dashed line if one is a subspace of the other.}\label{fig:subspaces}
\end{figure}

\begin{remark}
The subspace metric of \cite{Kotter.Kschischang2008} treats insertions and deletions of dimensions (called in \cite{Kotter.Kschischang2008} ``errors'' and ``erasures'', respectively) symmetrically. However, depending upon the position of the adversary in the network (namely, if there is a source-destination min-cut between the adversary and the destination) then a single error packet may cause the replacement of a dimension (i.e., a simultaneous ``error'' and ``erasure'' in the terminology of \cite{Kotter.Kschischang2008}). The injection distance, which is designed to ``explain'' a received subspace with as few error-packet injections as possible, properly accounts for this phenomenon, and hence the corresponding decoder produces a different result than a minimum subspace distance decoder. If it were possible to restrict the adversary so that each error-packet injection would only cause either an \emph{insertion} or a \emph{deletion} of a dimension (but not both), then the subspace distance of \cite{Kotter.Kschischang2008} would indeed be appropriate. However, this is not the model considered here.
\end{remark}

Let us now discuss an important fact about the subspace distance for general subspace codes (assuming for simplicity that $\rho=0$). The packet error correction capability of a minimum-subspace-distance decoder, $t_S$, is not necessarily equal to $\lfloor (\ds(\calC)-1)/2 \rfloor$ or $\lfloor (\ds(\calC)-2)/4 \rfloor$, but lies somewhere in between. For instance, in the case of a constant-dimension code $\Omega$, we have
\begin{align}
    \di(\mathcal{V},\mathcal{V}') &= \frac{1}{2}\ds(\mathcal{V},\mathcal{V}'),\quad \forall \mathcal{V},\mathcal{V}' \in \Omega,  \nonumber \\
    \di(\Omega) &= \frac{1}{2}\ds(\Omega). \nonumber
\end{align}
Thus, Theorem~\ref{thm:noncoherent-complete-capability} implies that $t_S = \lfloor (\ds(\calC)-2)/4 \rfloor$ exactly. In other words, in this special case, the approach in \cite{Kotter.Kschischang2008} coincides with that of this paper, and Theorem~\ref{thm:noncoherent-complete-capability} provides a converse that was missing in \cite{Kotter.Kschischang2008}. On the other hand, suppose $\Omega$ is a subspace code consisting of just two codewords, one of which is a subspace of the other. Then we have precisely $t_S = \lfloor (\ds(\calC)-1)/2 \rfloor$, since $t_S+1$ packet-injections are needed to get past halfway between the codewords.

Since no single quantity is known that perfectly describes the packet error correction capability of the minimum-subspace-distance decoder (\ref{eq:subspace-decoding-rule}) for general subspace codes, we cannot provide a definitive comparison between decoders (\ref{eq:subspace-decoding-rule}) and (\ref{eq:noncoherent-decoding-rule}). However, we can still compute bounds for codes that fit into Example~\ref{ex:decoding-comparison}.

\begin{example}
  Let us continue with Example~\ref{ex:decoding-comparison}. Now, we adjoin another codeword $\mathcal{V}_3 = \linspan{X_3}$
such that $\ds(\mathcal{V}_1,\mathcal{V}_3) = d$ and where
$\gamma' \triangleq \dim \mathcal{V}_3 - \dim \mathcal{V}_1$ satisfies $d/3 < \gamma' < d/2$. Also we assume that $\ds(\mathcal{V}_2,\mathcal{V}_3)$ is sufficiently large so as not to interfere with the problem (e.g., $\ds(\mathcal{V}_2,\mathcal{V}_3) > 3d/2$).

Let $t_{\scriptscriptstyle \rm S}$ and $t_{\scriptscriptstyle \rm M}$ denote the packet error correction capabilities of the decoders (\ref{eq:subspace-decoding-rule}) and (\ref{eq:noncoherent-decoding-rule}), respectively. From the argument of Example~\ref{ex:decoding-comparison}, we get $t_{\scriptscriptstyle \rm M} \geq \max\{\epsilon,\epsilon'\}$, while $t_{\scriptscriptstyle \rm S} < \min\{\epsilon,\epsilon'\}$, where $\epsilon' = (d-\gamma')/2$. By choosing $\gamma \approx d/3$ and $\gamma' \approx d/2$, we get $\epsilon \approx d/3$ and $\epsilon' \approx d/4$. Thus, $t_{\scriptscriptstyle \rm M} \geq (4/3) t_{\scriptscriptstyle \rm S}$, i.e., we obtain a 1/3 increase in error correction capability by using the decoder (\ref{eq:noncoherent-decoding-rule}).
\end{example}

\section{Conclusion}
\label{sec:conclusion}

We have addressed the problem of error correction in network coding under a worst-case adversarial model. We show that certain metrics naturally arise as the fundamental parameter describing the error correction capability of a code; namely, the rank metric for coherent network coding, and the injection metric for noncoherent network coding. For coherent network coding, the framework based on the rank metric essentially subsumes previous analyses and constructions, with the advantage of providing a clear separation between the problems of designing a feasible network code and an error-correcting outer code. For noncoherent network coding, the injection metric provides a measure of code performance that is more precise, when a non-constant-dimension code is used, than the so-called subspace metric.
The design of general subspace codes for the injection metric, as well as the derivation of bounds, is left as an open problem for future research.

\appendices

\section{Detection Capability}
\label{sec:detection-capability}

When dealing with communication over an adversarial channel, there is little justification to consider the possibility of error detection. In principle, a code should be designed to be unambiguous (in which case error detection is not needed); otherwise, if there is any possibility for ambiguity at the receiver, then the adversary will certainly exploit this possibility, leading to a high probability of decoding failure (detected error). Still, if a system is such that (a) sequential transmissions are made over the same channel, (b) there exists a feedback link from the receiver to the transmitter, and (c) the adversary is not able to fully exploit the channel at all times, then it might be worth using a code with a lower correction capability (but higher rate) that has some ability to detect errors.

Following classical coding theory, we consider error detection in the presence of a bounded error-correcting decoder. More precisely, define a \emph{bounded-discrepancy decoder with correction radius $t$}, or simply a \emph{$t$-discrepancy-correcting decoder}, by
\begin{equation}\nonumber
  \hat{x}(y) = \begin{cases}
    x & \text{if $\Delta(x,y) \leq t$ and $\Delta(x',y) > t$ for all $x' \neq x$, $x' \in \calC$} \\
    f & \text{otherwise}.
  \end{cases}
\end{equation}
Of course, when using a $t$-discrepancy-correcting decoder, we implicitly assume that the code is $t$-discrepancy-correcting. The \emph{discrepancy detection capability} of a code (under a $t$-discrepancy-correcting decoder) is the maximum value of discrepancy for which the decoder is guaranteed not to make an undetected error, i.e., it must return either the correct codeword or the failure symbol $f$.

For $t \in \mathbb{N}$, let the function $\sigma^t\colon \calX \times \calX \to \mathbb{N}$ be given by
\begin{equation}\label{eq:function-exact-detection}
 \sigma^t(x,x') \triangleq \min_{y\in \calY\colon \Delta(x',y)}\, \Delta(x,y) - 1.
\end{equation}

\begin{proposition}
 The discrepancy-detection capability of a code $\calC$ is given exactly by $\sigma^t(\calC)$. That is, under a $t$-discrepancy-correction decoder, any discrepancy of magnitude $s$ can be detected if and only if $s \leq \sigma^t(\calC)$.
\end{proposition}
\begin{proof}
 Let $t < s \leq \sigma^t(\calC)$. Suppose that $x \in \calX$ is transmitted and $y \in \calY$ is received, where $t < \Delta(x,y) \leq s$. We will show that $\Delta(x',y) > t$, for all $x' \in \calC$. Suppose, by way of contradiction, that $\Delta(x',y) \leq t$, for some $x' \in \calC$, $x' \neq x$. Then $\sigma^t(\calC) \leq \Delta(x,y) - 1 \leq s-1 < s \leq \sigma^t(\calC)$, which is a contradiction.

 Conversely, assume that $\sigma^t(\calC) < s$, i.e., $\sigma^t(\calC) \leq s-1$. We will show that an undetected error may occur. Since $\sigma^t(\calC) \leq s-1$, there exist $x,x' \in \calC$ such that $\sigma^t(x,x') \leq s-1$. This implies that there exists some $y \in \calY$ such that $\Delta(x',y) \leq t$ and $\Delta(x,y) - 1 \leq s-1$. By assumption, $\calC$ is $t$-discrepancy-correcting, so $\hat{x}(y) = x'$. Thus, if $x$ is transmitted and $y$ is received, an undetected error will occur, even though $\Delta(x,y) \leq s$.
\end{proof}

The result above has also been obtained in \cite{Yang++2008:WeightProperties}, although with a different notation (in particular, treating $\sigma^0(x,x')+1$ as a ``distance'' function). Below, we characterize the detection capability of a code in terms of the $\Delta$-distance.

\begin{proposition}
 For any code $\calC$, we have $\sigma^t(\calC) \geq \delta(\calC) -t -1$.
\end{proposition}
\begin{proof}
 For any $x,x' \in \calX$, let $y \in \calY$ be a solution to the minimization in (\ref{eq:function-exact-detection}), i.e., $y$ is such that $\Delta(x',y) \leq t$ and $\Delta(x,y) = 1 + \sigma^t(x,x')$. Then $\delta(x,x') \leq \Delta(x,y) + \Delta(x',y) \leq 1+\sigma^t(x,x') + t$, which implies that $\sigma^t(x,x') \leq \delta(x,x') - t - 1$.
\end{proof}

\begin{theorem}
 Suppose that $\Delta(\cdot,\cdot)$ is \magic. For every code $\calC \subseteq \calX$, we have $\sigma_t(\calC) = \delta(\calC) - t - 1$.
\end{theorem}
\begin{proof}
 We just need to show that $\sigma^t(\calC) \leq \delta(\calC) - t - 1$. Take any $x,x' \in \calX$. Since $\Delta(\cdot,\cdot)$ is \magic, there exists some $y \in \calY$ such that $\Delta(x',y) = t$ and $\Delta(x,y) = \delta(x,x') - t$. Thus, $\sigma^t(x,x') \leq \Delta(x,y)-1 = \delta(x,x') - t - 1$.
\end{proof}

\section{Proof of Lemma~\ref{lem:min-rank-quantity}}
\label{sec:proof-lemma}

First, we recall the following useful result shown in \cite[Proposition 2]{Silva++2008}. Let $X,Y \in \Fq^{N \times M}$. Then
  \begin{equation}\label{eq:rank-difference-bound}
    \rank (X-Y) \geq \max\{\rank X,\,\rank Y\} - \dim(\linspan{X} \cap \linspan{Y}).
  \end{equation}

\begin{proof}[Proof of Lemma~\ref{lem:min-rank-quantity}]
Using (\ref{eq:rank-difference-bound}) and (\ref{eq:bound-rank-product}), we have
\begin{align}
\rank (AX-BY)
&\geq \max\{\rank AX,\,\rank BY\} - \dim(\linspan{AX} \cap \linspan{BY}) \nonumber \\
&\geq \max\{\rank X - \rho,\,\rank Y -\sigma\} - \dim(\linspan{X} \cap \linspan{Y}). \nonumber
\end{align}

We will now show that this lower bound is achievable. Our approach will be to construct $A$ as $A = A_1 A_2$, where $A_1 \in \Fq^{L \times (L+\rho)}$ and $A_2 \in \Fq^{(L+\rho) \times n}$ are both full-rank matrices. Then (\ref{eq:bound-rank-product}) guarantees that $\rank A \geq n-\rho$. The matrix $B$ will be constructed similarly: $B = B_1 B_2$, where $B_1 \in \Fq^{L \times (L+\sigma)}$ and $B_2 \in \Fq^{(L+\sigma) \times N}$ are both full-rank.

Let $k = \rank X$, $s = \rank Y$, and $w = \dim (\linspan{X} \cap \linspan{Y})$. Let $W \in \Fq^{w \times m}$ be such that $\linspan{W} = \linspan{X} \cap \linspan{Y}$, let $\tilde{X} \in \Fq^{(k-w) \times m}$ be such that $\linspan{W} + \slinspan{\tilde{X}} = \linspan{X}$ and let $\tilde{Y} \in \Fq^{(s-w) \times m}$ be such that $\linspan{W} + \slinspan{\tilde{Y}} = \linspan{Y}$. Then, let $A_2$ and $B_2$ be such that
\begin{equation}\nonumber
A_2 X = \mat{W \\ \tilde{X} \\ 0} \quad \text{ and } \quad B_2 Y = \mat{W \\ \tilde{Y} \\ 0}.
\end{equation}

Now, choose any $\bar{A} \in \Fq^{i \times (k-w)}$ and $\bar{B} \in \Fq^{j \times (s-w)}$ that have full row rank, where $i = [k-w-\rho]^+$ and $j=[s-w-\sigma]^+$. For instance, we may pick $\bar{A} = \mat{I & 0}$ and $\bar{B} = \mat{I & 0}$. Finally, let
\begin{equation}\nonumber
  A_1 = \mat{I & 0 & 0 & 0 \\ 0 & \bar{A} & 0 & 0 \\ 0 & 0 & I & 0} \quad \text{ and } \quad B_1 = \mat{I & 0 & 0 & 0 \\ 0 & \bar{B} & 0 & 0 \\ 0 & 0 & I & 0}
\end{equation}
where, in both cases, the upper identity matrix is $w \times w$.

We have
\begin{align}
\rank (AX - BY)
&= \rank (A_1 A_2 X - B_1 B_2 Y) \nonumber \\
&= \rank (\mat{W \\ \bar{A} \tilde{X} \\ 0} - \mat{W \\ \bar{B} \tilde{Y} \\ 0}) \nonumber \\
&= \max\{i,\,j\} \nonumber \\
&= \max\{k-w-\rho,\,s-w-\sigma,\,0\} \nonumber \\
&= \left[\max\{k-\rho,\,s-\sigma\}-w\right]^+. \nonumber
\end{align}
\end{proof}

\section*{Acknowledgement}
  The authors would like to thank the anonymous reviewers for their helpful comments.

\clearpage

\begin{IEEEbiographynophoto}{Danilo Silva}
(S'06) received the B.Sc. degree from the Federal University of Pernambuco, Recife, Brazil, in 2002, the M.Sc. degree from the Pontifical Catholic University of Rio de Janeiro (PUC-Rio), Rio de Janeiro, Brazil, in 2005, and the Ph.D. degree from the University of Toronto, Toronto, Canada, in 2009, all in electrical engineering.

He is currently a Postdoctoral Fellow at the University of Toronto. His research interests include channel coding, information theory, and network coding.
\end{IEEEbiographynophoto}

\begin{IEEEbiographynophoto}{Frank R. Kschischang}
(S'83--M'91--SM'00--F'06) received the B.A.Sc. degree (with honors) from the University of British Columbia, Vancouver, BC, Canada, in 1985 and the M.A.Sc. and Ph.D. degrees from the University of Toronto, Toronto, ON, Canada, in 1988 and 1991, respectively, all in electrical engineering.

He is a Professor of Electrical and Computer Engineering and Canada Research Chair in Communication Algorithms at the University of Toronto, where he has been a faculty member since 1991. During 1997--1998, he was a Visiting Scientist at the Massachusetts Institute of Technology, Cambridge, and in 2005 he was a Visiting Professor at the ETH, Z{\"u}rich, Switzerland. His research interests are focused on the area of channel coding techniques.

Prof. Kschischang was the recipient of the Ontario Premier's Research Excellence Award. From 1997 to 2000, he served as an Associate Editor for Coding Theory for the \textsc{IEEE Transactions on Information Theory}. He also served as Technical Program Co-Chair for the 2004 IEEE International Symposium on Information Theory (ISIT), Chicago, IL, and as General Co-Chair for ISIT 2008, Toronto.
\end{IEEEbiographynophoto}

\end{document}